\documentclass{article}

\usepackage{arxiv}
\usepackage[utf8]{inputenc} 
\usepackage[T1]{fontenc}    
\usepackage{hyperref}       
\usepackage{url}            
\usepackage{booktabs}       
\usepackage{amsfonts}       
\usepackage{nicefrac}       
\usepackage{microtype}      
\usepackage{lipsum}		
\usepackage{graphicx}
\usepackage{natbib}
\usepackage{doi}

\usepackage{mathtools}
\usepackage{amssymb}
\usepackage{graphicx}
\usepackage{epstopdf}
\usepackage{times}
\usepackage{epsfig}
\usepackage{subfigure}
\usepackage[utf8]{inputenc}
\usepackage{amsthm}
\usepackage{algorithm}
\usepackage{algorithmic}
\usepackage{xcolor}
\usepackage{booktabs}
\usepackage{multirow}

\newtheorem{definition}{\bf Definition}
\newtheorem{lemma}{\bf Lemma}
\newtheorem{theorem}{\bf Theorem}

\newcommand{\ra}[1]{\renewcommand{\arraystretch}{#1}}

\title{A General Framework for Charger Scheduling Optimization Problems}

\date{January 15, 2020}	

\author{ {\hspace{1mm}Xual Li} \\
	Center for Advanced Computer Studies\\
	School of Computing and Informatics\\
	University of Louisiana at Lafayette\\
	Lafayette, LA 70503 \\
	\texttt{xuan.li1@louisiana.edu } \\
	\And
	{\hspace{1mm}Miao Jin} \\
	Center for Advanced Computer Studies\\
	School of Computing and Informatics\\
	University of Louisiana at Lafayette\\
	Lafayette, LA 70503 \\
	\texttt{miao.jin@louisiana.edu} \\
}

\renewcommand{\shorttitle}{\textit{arXiv} Template}

\begin{document}
\maketitle

\begin{abstract}
This paper presents a general framework to tackle a diverse range of NP-hard charger scheduling problems, optimizing the trajectory of mobile chargers to prolong the life of Wireless Rechargeable Sensor Network (WRSN),  a system consisting of sensors with rechargeable batteries and mobile chargers. Existing solutions to charger scheduling problems require problem-specific design and a trade-off between the solution quality and computing time. Instead, we observe that instances of the same type of charger scheduling problem are solved repeatedly with similar combinatorial structure but different data. We consider searching an optimal charger scheduling as a trial and error process, and the objective function of a charging optimization problem as reward, a scalar feedback signal for each search. We propose a deep reinforcement learning-based charger scheduling optimization framework. The biggest advantage of the framework is that a diverse range of domain-specific charger scheduling strategy can be learned automatically from previous experiences. A framework also simplifies the complexity of algorithm design for individual charger scheduling optimization problem. We pick three representative charger scheduling optimization problems, design algorithms based on the proposed deep reinforcement learning framework, implement them, and compare them with existing ones. Extensive simulation results show that our algorithms based on the proposed framework outperform all existing ones.
\end{abstract}


\section{Introduction}\label{sec:intro}
Wireless Rechargeable Sensor Network (WRSN), consisting of a group of sensors with rechargeable batteries and one or multiple mobile chargers, has become a promising solution to the energy limitation problem in Wireless Sensor Networks (WSNs). However, inefficient path planning of chargers may result in not just the waste of energy but also the death of nodes and failure of network tasks, considering a mobile charger needs to travel close to a sensor node to charge. Therefore, charger scheduling optimization has become a popular research topic that focuses on optimizing the trajectory of a mobile charger to prolong the life of a WRSN system.

Charger scheduling optimization problems can be roughly classified into two groups. One provides a  budget including the total charging time or energy spent on both charging and traveling. A charger seeks a path to maximize some objective function including the total energy charged on nodes  or number of nodes charged within the budget~\cite{liang2017approximation,chen2016charge}. The other requires a number of sensor nodes to be charged. A charger seeks a path to minimize some objective function including the total energy spent on the road, or the travel distance, or the total charging time~\cite{shi2011renewable,xie2012renewable,he2014mobile,lu2015wireless,xuan2018}.

Charger scheduling optimization problems are in general NP-hard. They are tackled by exact, approximation, and heuristic algorithms. Specifically, an exact algorithm with provable optimality guarantee uses enumeration strategy but fails when the size of dataset increases.  An approximation algorithm with provable solution quality is in general desirable. However, the computational complexity of an approximation algorithm may also go sky-high when the size of dataset or requirement of the solution quality increase. A heuristic algorithm is another alternative, fast to compute although lack of optimality or quality guarantee. Overall,  existing solutions require problem-specific design and a trade-off between the solution quality requirement and computing time.  

It is obvious that existing algorithm design on charger scheduling optimization problems fails to exploit the common characteristic of these problems - instances of the same type of problem are solved repeatedly with similar combinatorial structure but different data \cite{bello2016neural, khalil2017learning}. Instead, we consider searching an optimal charger scheduling as a trial and error process, and the objective function of a charging optimization problem as reward, a scalar feedback signal for each search. Specifically, we model charger scheduling optimization problems with a weighted graph and consider the objective function as a cumulative reward of charging sensor nodes along a  path in one cycle. We then build a deep reinforcement learning - a fundamental approach combined with deep neural network to optimize strategy to maximize the cumulative reward - based charger scheduling optimization framework. The biggest advantage of the framework is that a diverse range of domain-specific charger scheduling strategy can be learned automatically from previous experiences, i.e., different graphs with various sizes. A framework also simplifies the complexity of algorithm design for individual charger scheduling optimization problem.

We introduce the framework with four steps:

\textbf{Graph construction: } We use a weighted graph to model charger scheduling optimization problems.

\textbf{Graph representation: } We apply the structure2vec technique~\cite{dai2016discriminative,khalil2017learning} to represent the weighted graph.

\textbf{Deep reinforcement learning based framework: } The framework includes all the key components of a typical reinforcement learning algorithm and functions specifically designed for charger scheduling optimization problems.

\textbf{Deep-Q-Network algorithm:} A  Deep-Q-Network (DQN) is applied to learn the policy, i.e. an optimal charger scheduling strategy.

The rest of this paper is organized as follows: Section~\ref{Sec:Related} gives a brief review of charger scheduling optimization problems and introduction of the basic concepts of deep reinforcement learning. Sections~\ref{Sec:Problem1},~\ref{Sec:Problem2}, and~\ref{Sec:Problem3} introduce  three representative charger scheduling optimization problems, respectively, where the one in~\ref{Sec:Problem3} has not been studied yet. Section~\ref{Sec:Framework} provides the detail of proposed deep reinforcement learning-based charger scheduling optimization framework and shows in steps the algorithms designed to tackle the three charging problems based on the proposed framework. Section~\ref{Sec:Simulation} presents the simulation results and performance comparison of the algorithms designed based on the proposed framework and existing ones. Section~\ref{Sec:Conclusion} concludes the paper.

\section{Related Works}
\label{Sec:Related}

\subsection{Charger Scheduling Optimization}

Charger scheduling optimization has been a popular research topic that focuses on optimizing the trajectory of a mobile charger to prolong the life of a WRSN system. Charger scheduling optimization problems can be roughly classified into two groups.

The first group is that a charger with given budget (e.g., the total charging time, or the total energy spent on charging and traveling) seeks a path to achieve a maximum objective function (e.g., the total energy charged to nodes, or the total number of nodes being charged) \cite{liang2017approximation,chen2016charge}.  In~\cite{liang2017approximation}, a charger with energy bound seeks a path to maximize the charging rewards. The authors consider two scenarios - sensors can be charged to full capacity at one time or a certain energy level by several times - and provide a $4$-approximation algorithm. In~\cite{chen2016charge}, a charger seeks a charging path to maximize the number of mobile sensor nodes charged within a given charging time. The authors discretize the moving trajectory of each sensor node by the time and provide a quasi-polynomial time approximation algorithm.

The second group is that a charger is required to charge a given set of sensor nodes and seeks a path to achieve a minimum objective function (e.g., the total energy spent on the road, the total travel distance, or the total charging time)  \cite{shi2011renewable,xie2012renewable,he2014mobile,lu2015wireless,xuan2018,ma2018charging,xu2020minimizing}.  A solution may not exist for problems of this category. In~\cite{shi2011renewable}, the authors consider a  scenario where a mobile charger periodically travels inside a network to charge each sensor node. To
maximize the ratio of the charger's vacation time over the cycle time, the authors prove that its optimal traveling path in each renewable cycle is the shortest Hamiltonian cycle and  propose a heuristic solution. In~\cite{xie2012renewable}, the authors apply  multi-node wireless energy transfer technology and extend the study in~\cite{shi2011renewable} to a scenario where multiple nodes can be charged at the same time.  In~\cite{he2014mobile}, the authors introduce a  tree-based charging schedule to minimize the travel distance of a charger in robotic sensor networks. To guarantee the charging schedule depletion free for any robot, they provide theoretical guidance on the setting of remaining energy threshold at which  robots request energy replenishment. In~\cite{lu2015wireless}, the authors consider minimizing both the travel distance of a charger and the charging delay of sensor nodes as a set of nested Traveling Salesman Problems.  In~\cite{ma2018charging}, a mobile charger charges multiple sensors simultaneously to minimize the travel distance. In \cite{xu2020minimizing}, the authors consider a multi-node charging setting where multiple neighboring sensors can be charged simultaneously by a single charger. They provide an approximation algorithm to minimize the longest delay of sensors waiting for charge.

We pick three representative problems that cover the two major groups and design issues in WRSN in Sections~\ref{Sec:Problem1},~\ref{Sec:Problem2}, and~\ref{Sec:Problem3}, respectively. We use them as concrete examples to illustrate how our deep reinforcement learning-based framework can be applied and provide solutions with superior performance in the following sections. 

\subsection{Deep Reinforcement Learning}
\label{Sec:DRL}
Given a goal-directed agent in an uncertain environment, the agent  interacts with the environment through observations, actions, and feedback (rewards) on actions \cite{sutton1998reinforcement}. At each time step $t$, the agent observes current state $s_t$, and chooses action $a$. Then the state of the environment transits to next one $s_{t+1}$ and the agent receives a reward $r_t$.  $T(s_{t+1}|s_t,a)$ is a transition probability function indicating the probability that the environment will transfer to state $s_{t+1}$ if the agent take action $a$ at state $s_{t}$.  The agent learns through previous experiences to select an action and make the expected cumulative discounted rewards $E[\sum_{t=0}^{\inf} \gamma r_t]$ in the future maximized, where $\gamma$ is the discount rate between $0$ and $1$.

The agent takes action $A=\{a_1, ..., a_n\}$ at state $S=\{s_1, ..., s_m\}$ based on a policy denoted by $\pi(S, A)$. A policy can be either deterministic or stochastic. If the action space is discrete and the policy is deterministic, we choose value-based reinforcement learning, e.g. Q-learning. If the action space is continuous and the policy is stochastic, we then choose policy-based reinforcement learning, e.g. policy-gradient.  Considering the action space of charger scheduling optimization problems is discrete, i.e., a charger decides which exact node to charge next, we choose the value-based reinforcement learning technique, Q-learning, to build the framework.

Q-learning is a model-free reinforcement learning that does not require an agent with full knowledge of the whole environment. The agent simply maintains a Q-table, which stores Q-value where Q stands for quality/reward of each state-action pair. The agent will select the action that maximizes Q in the current state. However, it is infeasible to learn all the state-action pairs for most practical problems. Therefore, function approximation technique \cite{hornik1991approximation} is commonly used. For Q-learning, a function approximator $Q(S, A;\Theta)$ is parameterized by $\Theta$ with size much smaller than the combination of all possible state-action pairs. Deep Q-Network (DQN) \cite{riedmiller2005neural} applies deep neural networks as function approximators, combined with different techniques including the experience replay method \cite{mnih2013playing}. Considering the exponentially increased size of station-action pairs in charger scheduling optimization problems, we choose DQN to build the framework for charger scheduling optimization problems in WRSNs.

\section{Mobile Network Charging Path Optimization Problem}
\label{Sec:Problem1}

A battery-powered mobile sensor usually takes a pre-designed mobility pattern  to do the jobs such as scientific exploration. The mobile network charging path optimization problem studied in~\cite{chen2016charge} seeks an optimal path for a mobile charger to charge mobile sensor nodes as many as possible within a fixed time horizon.

\subsection{Network Model}
The network model in~\cite{chen2016charge} assumes a set of battery-powered mobile sensor nodes  that need to be recharged periodically. Denote  $V = \{v_i | 1 \le i \le n\}$ as the set of nodes deployed over a planar Field of Interest (FoI) and $p_i(t) (1 \le i \le n)$ as the position of  $v_i$ at time $t$. A sensor node $v_i$ is equipped with a rechargeable battery with capacity $B$.  The trajectory of a mobile sensor node can be a curve of any form without imposing any constraint on its mobility pattern, but its trajectory (i.e., the moving position $p_i(t)$ for node $v_i$ at any $t$) is known to the charger. A mobile sensor node remains stationary when being charged.

\subsection{Problem Formulation}
Assume the moving speed of a mobile charger is faster than the upper-bound of the moving speed of a mobile sensor. A mobile charger travels from a starting point, charges mobile sensor nodes one by one to a required battery level denoted as $\alpha$ before returning to the end point. The mobile charger needs to maximize the charged sensor nodes within a maximum charging timespan denoted as $C$.  The total charging time along the charging path $P$ is denoted as $\Gamma(P)$, including the traveling time of the charger,  the charging time of the sensor nodes, and the total sojourn time when the charger stays at a position without charging any mobile sensor. The total charging time $\Gamma(P)$ needs to be less than the maximum charging  timespan $C$.

The mobile network charging path optimization problem is defined as the following.
\begin{definition}
Given a required battery level $\alpha$ and a maximum charging timespan $C$, the mobile network charging path optimization problem is to schedule an optimal  path:
\begin{equation}\label{eq:equation_mobile}
 P^* = \arg\max\limits_{\Gamma(P)\le C} |\Lambda(P)|
\end{equation}
\end{definition}
where $\Lambda(P)$ is the set of nodes charged on the path $P$.

\subsection{Problem Hardness}
The mobile network charging path optimization problem is APX-hard.

\begin{theorem}{}
The mobile network charging path optimization problem is APX-hard, or NP-hard.
\end{theorem}

\begin{proof}
The proof can be found in the reference~\cite{chen2016charge}.
\end{proof}

\subsection{Algorithm}

A quasi-polynomial time algorithm that achieves a poly-logarithmic approximation is introduced in~\cite{chen2016charge} .  The trajectory of each mobile sensor node is discretized by a time step  $\Delta t$ to construct a directed acyclic graph. Vertices of the graph represent the discretized points along the trajectories of mobile sensor nodes, the starting and ending positions of the charger. The approximation algorithm recursively decomposes the  problem of searching an optimal charging path into sub-problems of searching sub-paths.

The computational complexity of the  algorithm is $O(n_d\min(n_d, C)\log C)^L$, where $n_d$ is the size of nodes after discretization, $C$ is the maximum charging timespan, and $L$ is the recursion level of the algorithm. Let $n$ represent the original size of mobile sensor nodes, then $n_d = n \frac{C}{\Delta t}$. For a network with large size $n$ and a long charging timespan $C$, it is obvious that the approximation algorithm has to sacrifice the solution quality by either increasing the time step size $\Delta t$ or decreasing the recursion level $L$.

\section{Fully Charging Reward Maximization Problem}
\label{Sec:Problem2}

Given a set of energy-critical sensors and a mobile charger with a fixed energy capacity,  the fully charging reward maximization problem studied in~\cite{liang2017approximation} seeks an optimal  tour for the mobile charger to  maximize the sum of charging energy rewards.

\subsection{Network Model}

The network model in~\cite{liang2017approximation}  assumes a set of stationary sensor nodes,  $V = \{v_i | 1 \le i \le n\}$, deployed over a planar FoI with locations, $P = \{p_i | 1 \le i \le n\}$. Each sensor node $v_i$ is equipped with a rechargeable battery with a capacity of $B$.

\subsection{Problem Formulation}

A mobile charger with an energy capacity $IE$ sends its departure time denoted as $t_0$ from the service station to each sensor node. When receiving the message, node $v_i$ estimates its residual energy at $t_0$ denoted as $B_i(t_0)$. If $B_i(t_0)$ is less than an energy threshold $\alpha$, i.e., $B_i(t_0)/B\le\alpha$, $v_i$ calculates a positive integer number $\pi_i(1 \le \pi_i\le n^2)$ that models the gain/reward if being charged. The reward is proportional to the amount of  energy required to charge $v_i$ to its full capacity.  Node $v_i$ then sends a charging request $(id,p_i,\pi_i,B_i(t_0))$,  including node $v_i$ ID, position $p_i$, charging reward $\pi_i$, and residual energy $B_i(t_0)$ at $t_0$, to the charger. 

The mobile charger seeks a charging path $P$ that maximizes the total  reward of nodes charged. Note that the total amount of energy consumed in charging sensor nodes and traveling along the $P$ should be no larger than the energy capacity $IE$ of the mobile charger.

The fully charging reward maximization problem is defined as the following.
\begin{definition} [Fully charging reward maximization problem]
The fully charging reward maximization problem is to schedule a charging path $P$

\begin{subequations}\label{eq:equation_fullycharging}
	\begin{align}
    &\max~ \sum_{v_i\in P} \pi_i \\
	&{\rm s.t.}~~ \sum_{v_i \in P}(B-B_i(t_0)) + \Omega(P)*\xi\le IE
	\end{align}
\end{subequations}

\end{definition}
where $\Omega(P)$ is the length of path $P$, and $\xi$ is the amount of energy consumption per traveling unit of the charger.

\subsection{Problem Hardness}
The fully charging reward maximization problem is NP-hard.

\begin{theorem}{}
The fully charging reward maximization problem is NP-hard.
\end{theorem}

\begin{proof}
The proof is contained in the reference \cite{liang2017approximation}.
\end{proof}

\subsection{Algorithm}
A $4$-approximation algorithm is introduced in~\cite{liang2017approximation}. The basic idea of the algorithm is to reduce the fully charging reward maximization problem to the orienteering problem with an approximation algorithm introduced in~\cite{bansal2004approximation}. The orienteering problem is to find a path $P$ in a graph $G$ from a starting vertex to an ending one such that the total reward collected from vertices along $P$, $\sum_{v_i\in P} \pi_i$ is maximized, and the length of $P$ is no longer than the length of a maximum path $\Omega(P_{\max})$.

The time complexity of the $4$-approximation algorithm is $T_{\rm ori}(3|V(G)|,2(|V(G)|+|E(G)|))$, where $T_{\rm ori}(|V(G)|,|E(G)|)$ is the time complexity of the approximation algorithm in~\cite{bansal2004approximation}.

\section{Optimal k-coverage Charging Problem}
\label{Sec:Problem3}
One basic requirement of WSN deployment is full area coverage of a field of interest (FoI). Multiple coverage, where each point of the FoI is covered by at least $k$ different sensors with $k>1$ ($k$-coverage), is often applied to increase the sensing accuracy of data fusion and enhance the fault tolerance in case of node failures~\cite{Yang06onconnected,Simon07dependablek,Handbook2008,VariableRadii2009,Multiple-Coverage2011,li2015autonomous}. A common practice to achieve $k$-coverage is a high density of sensor nodes randomly distributed  over the monitored FoI~\cite{Kumar:2004:KMS:1023720.1023735,Hefeeda:2007:RKA:2931311.2931587}.

Optimal k-coverage charging problem studies the mobile charger scheduling scenario in which the $k$-coverage ability of a network system needs to be maintained. A node sends a charging request with its position information and a charging deadline estimated based on its current residual energy and battery consumption rate.  A mobile charger seeks a path to charge sensor nodes before their charging deadlines under the constraint of maintaining the $k$-coverage ability of the monitored are. At the same time, the charger tries to maximize the energy usage efficiency, i.e., minimizing the energy consumption on traveling per tour. 

The Optimal k-coverage charging problem has not been studied yet. We describe the network model, formulate the problem, analyze its hardness, and propose a dynamic programming-based algorithm in detail in the following sections.

\subsection{Network Model}

We consider a set of stationary sensor nodes,  $V = \{v_i | 1 \le i \le n\}$, deployed over a planar FoI denote as $A$  with locations, $P = \{p_i | 1 \le i \le n\}$. For each sensor node $v_i$, we assume a disk sensing model with sensing range $r$. If the Euclidean distance between a point $q \in A $ and node position $p_i$ is within distance $r$, i.e., $||p_i - q||_{L_2} \le r$, then the point $q$ is covered by sensor node $v_i$, and we use $v_i(p) = 1$ to represent it, as shown in equation \eqref{eq:equation1}:
\begin{equation}
\label{eq:equation1}
v_i(q)=
\begin{cases}
1, & ||p_i - q||_{L_2} \le r \\
0, & \text{otherwise.} \\
\end{cases}
\end{equation}

\begin{definition}[Full Coverage]
	If for any point $q \in A$, there exists at least one sensor node covering it, i.e., $\sum_{i=1}^{n}  v_i(q) \ge 1$, then area $A$ is full covered.
\end{definition}

\begin{definition}[$K$-Coverage]
	If for any point $q \in A$, there exist at least $k \ge 1$ sensor nodes covering it, i.e., $\sum_{i=1}^{n} v_i(q) \ge k$, then area $A$ is $k$-covered.
\end{definition}

It is obvious that full coverage is a special case of $k$-coverage when $k=1$.

\subsection{Problem Formulation}

 A sensor node $v_i$ is equipped with a rechargeable battery with capacity $B$. $B_i(t)$ denotes  the residual energy of sensor node $v_i$ at time $t$. A mobile charger sends its departure time denoted as $t_0$ from service station to each sensor node. When receiving the message, node $v_i$ estimates its residual energy $B_i(t_0)$ at $t_0$. If it is less than an energy threshold $\alpha$, i.e., $B_i(t_0)/B\le\alpha$, sensor node $v_i$ sends a charging request $(id,p_i,D_i)$ including its ID, position $p_i$, and charging deadline $D_i$ to the charger. A charging deadline, i.e., an energy exhausted time, is estimated based on the residual energy $B_i(t_0)$ at $t_0$ and an average battery consumption rate denoted as $\beta_i$. Specifically, $D_i = B_i(t_0)/\beta_i$. Note that nodes may have different energy consumption rates.

A mobile charger with an average moving speed $s$ is responsible for selecting and charging a number of sensor nodes before their charging deadlines to maintain  $k$-coverage of area $A$, and it also seeks a path with a minimum energy consumption on traveling. We assume that the time spent on charging path is less then the operation time of sensors, so a sensor node only needs to be charged once in each tour. Unless under an extremely dense sensor deployment, we consider that a mobile charger charges sensor nodes one by one because the energy efficiency reduces dramatically with distance; the energy efficiency drops to 45 percent when the charging distance is $2$m ($6.56$ ft)~\cite{kurs2007power}. We denote $r_c$ the energy transfer rate of a mobile charger.

The charging time is defined as the following:

\begin{definition}[Charging Time] \label{def:charging}
	Denote $P$ a charging path and $t_P(v_i)$ the charging time along $P$ at node $v_i$. If $P$ goes from nodes $v_i$ to $v_j$, the charging time begins at node $v_j$ is

	\begin{equation}
	\label{eq:equation5}
	t_P(v_j) =
	\begin{cases}
	t_P(v_i) +\frac{B-B_i(t_P(v_i))}{r_c} + \frac{d_{ij}}{s},  \\
    \quad \quad \quad \text{if } t_P(v_i) + \frac{B-B_i(t_P(v_i))}{r_c} + \frac{d_{ij}}{s} \le D_j \\
	\text{inf},  \\
    \quad \quad \quad \text{otherwise}, \\
	\end{cases}
	\end{equation}

	where  $d_{ij}$ is the Euclidean distance between nodes $v_i$ and $v_j$ and  $s$ is the average moving speed of a mobile charger. The residual energy $B_i(t)$ is estimated as $B_i(t) = B_i(t_0) - \beta_i*(t-t_0)$.
\end{definition}

The optimal $k$-coverage charging problem is then be formulated as the following.

\begin{definition} [Optimal $k$-coverage charging problem]
Given a set of sensor nodes $V = \{v_i | 1 \le i \le n\}$, randomly deployed over a planar region A with locations $P = \{p_i | 1 \le i \le n\}$ such that every point of $A$ has been at least $k$ covered initially, the optimal $k$-coverage charging problem is to schedule a charging path $P$

\begin{subequations}\label{eq:equation2}
	\begin{align}
 &\min~ |P| \\
	&{\rm s.t.}~~ \sum_{i=1}^n v_i(q) \ge k, \forall q \in A.\label{eq:equation2b}\\
	&\quad\quad t_P(v_i)\leq D_i, \forall v_i \in P.
	\end{align}
\end{subequations}

\end{definition}
Note that a charger does not need to respond all the nodes sending requests.

\subsection{Problem Hardness}
\label{Sec:hardness}

The optimal $k$-coverage charging problem is NP-hard.

\begin{theorem}{}
The optimal $k$-coverage charging problem is NP-hard.
\end{theorem}

\begin{proof}
To prove the NP-hardness of  optimal $k$-coverage charging problem, we prove that an NP-hard one, traveling salesman problem with deadline that seeks a minimum traveling cost to visit all the nodes before their deadlines, can be reduced to the trivial case of the optimal $k$-coverage charging problem in polynomial time. 
	
We consider a trivial case of the optimal $k$-coverage charging problem: we require $k=1$ and assume that the initial deployment of sensor nodes has no coverage redundancy. A mobile charger needs to charge all the sensor nodes sending requests before their deadlines to maintain a full coverage. It is straightforward to see that a solution of traveling salesman problem with deadline is also a solution of the trivial case of the optimal $k$-coverage charging problem and vice versa. Since even finding a feasible path for traveling salesman problem with deadline is NP-complete \cite{savelsbergh1985local},  the optimal $k$-coverage charging problem is then NP-hard.
\end{proof}

\subsection{Dynamic Programming Based Algorithm}
\label{sec:area_segmentation}

\subsubsection{Area Segmentation and Time Discretization}
Considering that the sensing region of a sensor node $v_i$ is centered at $p_i$ with radius $r_i$, disk-shape sensing regions of a network divide a planar FoI $A$ into a set of subregions, marked as $A = \{a_i | 1 \le i \le m\}$. Then $\sum_{i=1}^n v_i(a_i)$ is the number of sensors with $a_i$ within their sensing ranges. In the definition of optimal $k$-coverage charging problem, we assume $\sum_{i=1}^n v_i(a_i) \ge k$ with the initial deployment of a network. Denote $r(a_i)$  the number of sensor nodes sending charging requests with sensing regions including subregion $a_i$. Three cases exist for subregion $a_i$:

\textbf{Case I: } $\sum_{i=1}^n v_i(a_i) - r(a_i) \ge k$: all the requests are not essential.

\textbf{Case II: } $\sum_{i=1}^n v_i(a_i)= k$ and $\sum_{i=1}^n v_i(a_i) - r(a_i) < k$: a mobile charger needs to charge all the sensor nodes sending requests with their sensing regions containing subregion $a_i$  before their charging deadlines.

\textbf{Case III: } $\sum_{i=1}^n v_i(a_i)> k$ and $\sum_{i=1}^n v_i(a_i) - r(a_i) < k$: a mobile charger needs to charge at least $ k - \sum_{i=1}^n v_i(a_i) + r(a_i) $ sensor nodes sending requests with their sensing regions containing subregion $a_i$  before their charging deadlines.

A table denoted as $T$ with size $m$, is constructed to store the minimum number of sensors to charge for each $a_i$. Specifically, $T[i] = k - \sum_{i=1}^n v_i(a_i) + r(a_i)$. If the value is negative, we simply set $T[i]$ to zero.

For a sensor node $v_i$ sending request, we divide its time window $[t_0, t_0+D_i]$ into a set of time units $\{t_i^k | 0 \le k \le D_i\}$, where $t_i^0 = t_0$ and $t_i^{D_i} = D_i$. We represent node $v_i$ with a set of discretized nodes $\{v_i(t_i^k) | 0 \le k \le D_i\}$, where $v_i(t_i^k)$ represents node $v_i$ at time $t_i^k$.

\subsubsection{Graph Construction} \label{sec:General_Graph_Construction}

We construct a directed graph denoted as $G$ with vertices and edges defined as follows.

\textbf{Vertices.} The vertex set $V(G)$ includes the discretized sensor nodes sending charging requests, i.e., $\{v_i(t_i^k) |  0 \le k \le D_i\}$.

\textbf{Edges.}
There exists a directed edge $\overrightarrow{v_i(t_i^{k})v_j(t_j^{k'})}$ from $v_i(t_i^{k})$ to $v_j(t_j^{k'})$ in the edge set $E(G)$ if and only if
\[
t_i^{k}+\frac{B-B_i}{r_c} + \frac{d_{ij}}{s} > t_j^{k'-1},
\]
\[
t_i^{k}+\frac{B-B_i}{r_c}+\frac{d_{ij}}{s}\le t_j^{k'},
\]
where $k' > 0$.

A directed edge ensures that a mobile charger arrives at sensor node $v_j$ before its charging deadline with the charging time being the arrival time of the charger.

\begin{theorem}{}
	$G$ is a directed acyclic graph (DAG).
\end{theorem}
\begin{proof}
	Suppose there exists a cycle in $G$. Assume vertices $v_i(t_i^k)$ and $v_j(t_j^{k'})$ are on the cycle. Along the directed path from $v_i(t_i^k)$ to $v_j(t_j^{k'})$, it is obvious that $t_i^k < t_j^{k'}$. However, along the directed path from $v_j(t_j^{k'})$ to $v_i(t_i^k)$, we have $t_j^{k'} < t_i^k$. Contradiction, so $G$ is a directed acyclic graph.
\end{proof}

\begin{definition}[Clique]
	A set of nodes $\{v_i(t_i^k) | 0 \le k \le D_i\}$ in $G$ is defined as a clique if they correspond to the same  node $v_i$ at different time units.
\end{definition}

\begin{definition}[Feasible Path]
	A path $P$ in $G$ is a feasible one if it passes no more than one vertex of a clique. At the same time, charging along $P$ satisfies the $k$-coverage requirement of the given network.
\end{definition}

\subsubsection{Dynamic Programming Algorithm}

We design a dynamic programming-based algorithm to find an optimal charging path for the $k$-coverage charging problem.

To make sure that the computed charging path passes no more than one discretized vertex of a sensor node, we apply the color coding technique introduced in~\cite{alon1995color} to assign each vertex a color. Specifically, we generate a coloring function $c_v: V \rightarrow \{1,...,n\}$ that assigns each sensor node a unique node color. Each sensor node then passes the node color to its discretized ones. A path in $G$ is said to be colorful if each vertex on it is colored by a distinct node color. It is obvious that a colorful path in $G$ passes no more than one discretized vertex of a sensor node.

To take into the consideration of traveling distance from service station to individual sensor node, we add an extra vertex denoted as $v_0$ and connect it with directed edges to vertices in $G$, i.e., $\{ v_i(t_i^k) | k = 0 \}$. The length of edge $\overrightarrow{v_0v_i(t_i^0)}$ is the Euclidean distance between the service station and sensor node $v_i$. The table $T$ constructed in Sec.~\ref{sec:area_segmentation} is stored at $v_0$.

We first topologically sort the new graph, i.e., $G + v_0$. Then we start from $v_0$ to find colorful paths by traversing from left to right in linearized order. Specifically,  $v_0$ checks neighbors connected with outgoing edges and sends table $T$ to those contributing to the decrease of at least one table entry. Once a vertex $v_i(t_i^0)$ receives $T$, $v_i(t_i^0)$ checks the subregions within its sensing range and updates the corresponding entries of $T$. $v_i(t_i^0)$ also generates a color set $C=\{ c(v_i(t_i^0)) \}$ and stores with $T$, which indicates a colorful path of length $|C|$.

Similarly,  suppose the algorithm has traversed to vertex $v_i(t_i^k)$, we check each color set $C$ stored at $v_i(t_i^k)$ and its outgoing edge $v_j(t_j^{k'})$. If $c(v_j(t_j^{k'})) \not\in C$ and charging $v_j(t_j^{k'})$ helps decrease at least one entry of $T$ associated with $C$, we add the color set $C = \{ C + c(v_j(t_j^{k'}))\}$ along with the updated $T$ to the collection of $v_j(t_j^{k'})$.

After the update of the last vertex  in linearized order, we check the stored $T$s in each node and identify those with all zero entries. A color set associated with a $T$ with all zero entries represents a colorful path that is a feasible solution of the $k$-coverage problem.

A path can be easily recovered  from a color set. The basic idea is to start from vertex $v_i(t_i^k)$ with a color set $C$. We check the stored color sets of vertices connected to $v_i(t_i^k)$ with incoming edges. Assume we identify a neighbor node $v_j(t_j^{k'})$ storing a color set $C-c(v_i(t_i^k))$, then we continue to trace back the path from $v_j(t_j^{k'})$ with a color set $C-c(v_i(t_i^k))$. When we trace back to $v_0$, we have recovered the whole charging path. Among all feasible charging paths, the one with a minimal traveling distance is the optimal one.

\begin{lemma}{}
	The algorithm returns an optimal solution of the $k$-coverage charging problem, i.e., a feasible path maximizing the energy usage efficiency, if it exists.
\end{lemma}
\begin{proof}
	We first show that the algorithm returns a feasible path. A path returned by the algorithm is a colorful one that guarantees the path passes a sensor node no more than once. In the meantime, charging time at each sensor node along the path is before its deadline, otherwise a directed edge along the path won't exist. Array $T$ with all zero entries makes sure that the $k$-coverage is maintained.
	
	When the algorithm has traversed to the $i^{th}$ node in linearized order, each colorful path passing through the node has been stored in the node.
\end{proof}

Note that the computational complexity of the dynamic programming algorithm can increase exponentially in the worst case because the stored color sets at a vertex can increase exponentially to the size of sensor nodes n.

\section{Charger Scheduling Optimization Framework}
\label{Sec:Framework}

We show the framework by four steps and use the three representative charger scheduling optimization problems discussed in Secs.~\ref{Sec:Problem1}, \ref{Sec:Problem2}, and \ref{Sec:Problem3} as concrete examples to illustrate the algorithm design based on the framework.

\subsection{Graph Construction}\label{Sec:Graph_Construction}

Given a set of sensor nodes represented by $V = \{v_i | 1 \le i \le n\}$ deployed over a planar region with the initial positions denoted by $P = \{p_i | 1 \le i \le n\}$, the sensing range of a sensor node $v_i$ is $r$ and the sensing model is a disk. Each node $v_i$ is equipped with a rechargeable battery of capacity $B$. When the  residual energy of $v_i$ represented as $B_i(t)$ at time $t$ is below a specific threshold, the sensor node $v_i$ will send a charging request to a mobile charger. The charger collects all charging requests before leaving from base station denoted by $v_0$ at $t_0$. The average speed of the charger is $s$.

We use weighted graph $G(V, E, \omega) $ to model the charger scheduling optimization problems.  Specifically,

\textbf{Vertices.}
A vertex  $v_i \in V(G)$ represents a sensor node sending charging request.

\textbf{Edges.}
An edge $e(v_i,v_j) \in E(G)$ indicates a possible charging path of a charger from sensor node $v_i$ to  $v_j$ without the violation of any constraints.

\textbf{Edge Weight.}
The  weight $\omega(v_i,v_j)$ assigned to edge $e(v_i,v_j)$ represents the Euclidean distance of nodes $v_i$ and $v_j$.

More constraints are added to the graph model for each charger scheduling optimization problem. Specifically,

\begin{itemize}
\item \textbf{Mobile Network Charging Path Optimization Problem:} The graph is a fully connected undirected one with an edge connecting every pair of sensor nodes. The edge weight changes dynamically because the sensor nodes are mobile ones. With the assumption that the exact location of each mobile node is known to a charger at any time $t$, we discretize the time and update the graph in each time step.

\item \textbf{Fully Charging Reward Maximization Problem:} The graph is a fully connected undirected one with an edge connecting every pair of sensor nodes. Each vertex $v_i$ is additionally associated with a positive integer with range $[1, n^2 ]$  to model the gain of charging $v_i$ by a mobile charger. A sensor with less residual energy is assigned a larger prize as it needs to be charged more urgently.

\item \textbf{Optimal k-coverage Charging Problem:} 

\textbf{Vertices.}
The vertex set $V(G)$ includes all the sensor nodes, i.e.,  $\{v_i | 1 \le i \le n\}$ and a start vertex denoted as $v_0$. Each vertex $v_i$ has a deadline $D_i$. The location of $v_0$ is the service station and the deadline of $v_0$ is $\text{inf}$. Note that the deadline of sensor nodes without sending any request is set as $0$.

\textbf{Edges.}
For any sensor nodes $v_i$ and $v_j$ in $V(G)$, $d_{ij}$ is the euclidean distance between $v_i$ and $v_j$. There exists an edge $\overrightarrow{v_iv_j}$ in $V(G)$ if and only if the inequality below holds
\begin{equation}
\label{eq:equation4}
\frac{B-B_i(t_0)}{r_c}+\frac{d_{ij}}{s}\le D_j
\end{equation}
where $B_i(t_0)$ denotes as the residual energy of sensor node $i$ at charger departure time $t_0$, $r_c$ is the energy transfer rate of the mobile charger and $s$ is its average moving speed. If a charging path $P$ goes from node $v_i$ to node $v_j$, the charging time beginning at node $v_j$ is given by Def.~\ref{def:charging}.

\begin{definition}[Feasible Path]
	A path $P$ in $G$ is a feasible one if it starts from and ends at $v_0$, does not traverse repeated vertex, and charges vertices before their charging deadlines. At the same time, charging along $P$ satisfies the $k$-coverage requirement of the given network.
\end{definition}

\end{itemize}

\subsection{Graph Representation}
As we have mentioned in Sec.~\ref{Sec:DRL}, deep Q-learning applies parametrized function $Q(S, A;\Theta)$ to approximate a large-sized state-action value function $Q(S,A)$ that includes the combination of all pairs of states and actions, where $Q(S, A;\Theta)$ is parameterized by $\Theta$ with size much smaller than $Q(S, A)$. It is expected that the state-action value function $Q(S, A;\Theta)$  has well summarized the state $S$ of the current problem solving, i.e., incorporating  current partial solution including the selected vertices and their order into the graph model constructed in Sec.~\ref{Sec:Graph_Construction}. The state-action value function $Q(S, A;\Theta)$ should also have an estimation of the reward value when a new vertex is picked considered as taking an action $A$ in current state $S$. 

However, it is challenging to accurately describe both $S$ and $A$ on a graph. They may depend on the global and local statistics of the current graph. There exist different graph representation techniques. We apply  structure2vec~\cite{dai2016discriminative,khalil2017learning}, a deep learning-based graph embedding technique, to compute a p-dimensional node embedding for each vertex and a feature vector with the same dimension for the graph. The state-action value function $Q(S, A;\Theta)$ can then be defined by the computed feature vector and node embedding parameterized by $\Theta$. Parameters $\Theta$ will be learned later using deep reinforcement learning algorithm.

\subsection{Deep Reinforcement Learning Framework}
States, actions, transition, rewards, and policy are key components of any typical reinforcement learning algorithm. We add objective function, insertion function, and stop function into the deep reinforcement learning based charging framework. Before we discuss these key components of the framework, we will briefly explain some terms borrowed from reinforcement learning. An episode in the framework refers to a complete sequence of  sensor nodes charged under constraints or requirements. A step, denoted by $t$, refers to charging a single sensor node in an episode.

\textbf{Objective function:} An objective function, denoted by $f$, reflects the quality of a partial solution to a charging problem. The definition of $f$ varies from one charging problem to another, e.t., from maximizing the number of charging nodes to minimizing the traveling distance.
To keep consistent, we turn $f$ to negative when the problem requires minimizing some value.

\textbf{Insertion function:} An insertion function, denoted by $g$, is designed to insert vertex $v$ to the best position in a partial solution $s_{t+1}$, e.g., $s_{t+1}:=g(s_t,v)$.

\textbf{Stop function:}  A stop function checks the problem constraints and determines when to stop the current episode.

\textbf{States:} A state $s_t$ is an ordered list of visited vertices at time step $t$, represented as a $p$-dimensional vector using the structure2vec~\cite{dai2016discriminative,khalil2017learning} graph representation technique. It is a partial solution $s_t\subseteq V(G)$. The initial state is denoted by $s_0 = (v_0)$. A stop function will determine the termination state $s_\text{end}$.

\textbf{Actions:} Actions include all the candidate vertices at state $s_t$. Such candidacy depends on the definition of the individual charging problem. We will discuss it later. Similar to a state, an action is represented as a $p$-dimension node embedding using the structure2vec~\cite{dai2016discriminative,khalil2017learning} graph embedding technique.

\textbf{Transition:} Suppose $v$ is a candidate vertex. After taking the action $v$ at time step $t$, state $s_t$ is transitioning to  $s_{t+1} := g(s_t, v)$ where $v$ is inserted  to the best position by an insertion function $g$.

\textbf{Rewards:} A reward function, denoted by $r(s_t, v)$, reflects the change of the objective function $f$ after taking action $v$ at state $s_t$ and transitioning to state $s_{t+1}$:
\begin{equation}
r(s_t,v) = f(s_{t+1}) - f(s_{t}).
\end{equation}
It is obvious that the cumulative reward of an episode equals the objective function of the termination state, i.e. $\sum_{i=1}^{n}r(s_i,v_i) = f(s_\text{end})$ assuming the episode takes $n$  steps.

\textbf{Policy:} We choose an epsilon greedy policy. The policy will choose an action achieving the current highest reward. However, with a small probability, it will instead randomly select the next action to prevent the stuck of a local optimal solution. Specifically, at time step $t$, an action $v_t$ is selected by
\begin{equation}\label{eq:epsilon_greedy}
v_t  =
\begin{cases}
\arg\max_{v\in\bar{s}_t}Q(s_t,v;\Theta),
& \text{with probability } 1-\epsilon\\
\text{select a random vertex } v_t\in\bar{s}_t,
& \text{otherwise}
\end{cases}
\end{equation}

We will illustrate how to tailor the model for each of the selected charger scheduling optimization problems.

\subsubsection{Mobile Network Charging Path Optimization Problem:} 

\hspace{0.35cm} \textbf{Objective function: } The objective function $f$  counts the number of selected vertices, i.e., charged mobile nodes. 

\textbf{Insertion function: } The Insertion function $g$  simply inserts the  selected vertex $v$ at time step $t$ to the end of state $s_t$.  

\textbf{Stop function: } The stop function checks  $s_{t+1}$ to make sure that the total charging time  is no larger than a given maximum timespan of the charger. Otherwise, the selected vertex $v$ will be removed.  

\textbf{Actions: }  Actions at time step $t$ include all non-selected vertices in current state $s_t$. 

\textbf{Rewards: } The reward function $r(s_t,v)$ is defined as
\begin{equation}
r(s_t, v)=1
\end{equation}
because inserting one vertex to a partial solution contributes the objective function $f$ by one count.

\subsubsection{Fully Charging Reward Maximization Problem:} 

\hspace{0.35cm} \textbf{Objective function: } The objective function $f$  counts the total prizes collected from selected vertices, i.e.,  mobile nodes charged. The prize at each vertex is proportional to the energy charged to its full capacity. 

\textbf{Insertion function: } The Insertion function $g$  finds a position to insert $v$ at $s_t$ that minimizes the energy of charger spent on road, which is defined as:
\begin{equation}
\label{eq:equation7}
g(s_t,v) =
\arg\min_i^{}{\{E_{0i}+E_{it}\}}
\end{equation}
where $E_{0i}+E_{it}$ is the energy of charger spent on path when charging vertex $v$ as the $i$-th one among the selected vertices. 

\textbf{Stop function: } The stop function checks  $s_{t+1}$ to make sure that the total amount of energy consumed on sensor charging and the traveling  is no greater than the energy capacity of a mobile charger . Otherwise, the selected vertex $v$ will be removed.  

\textbf{Actions: } Actions at time step $t$ include all non-selected vertices in current state $s_t$. 

\textbf{Rewards: } The reward function $r(s_t,v)$ is defined as
\begin{equation}
\label{eq:equation7}
r(s_t,v) = E_{v}
\end{equation}
where $E_{v}$ is the energy charged to get the full capacity of $v$.

\subsubsection{Optimal k-coverage Charging Problem:} 

\hspace{0.35cm} \textbf{Objective function: } The objective function $f$  counts the total traveling distance of a charger. To maximize $f$, we turn $f$ to negative. 

\textbf{Insertion function: } The Insertion function $g$ finds a position to insert $v$ at $s_t$ that maximizes the reward  and ensures  each node charged before its deadline.  

\textbf{Stop function: }: The stop function checks whether the current charging path guarantees the $k$-coverage requirement of the area. If the requirement is satisfied or  R(S,v) is -inf, the algorithm terminates.

\textbf{States: }  A state $S$ is a partial solution $S\subseteq V(G)$, an ordered list of visited vertices. The first vertex in $S$ is $v_0$.

\textbf{Actions: } Let $\bar{S}$ contain vertices not in $S$ and has at least one edge from vertices in $S$. An action is a vertex $v$ from $\bar{S}$ returning the maximum reward. After taking the action $v$, the partial solution $S$ is updated as
\begin{equation}
S' := (S, v), \text{where } v = \arg\max_{v\in\bar{S}}Q(S,v)
\end{equation}
$(S,v)$ denotes appending $v$ to the best position after $v_0$ in $S$ that introduces the least traveling distance and maintains all vertices in the new list a valid charging time.

\textbf{Rewards: } The reward function $R(S,v)$ is defined as the change of the traveling distance when taking the action $v$ and transitioning from the state $S$ to a new one $S'$. Assume $v_i$, $v_j$ are two adjacent vertex in $S$, $v_0$ is the first vertex in the $S$, and $v_t$ is the last vertex in the $S$.

The reward function $R(S,v_k)$ is defined as follows:
\begin{equation}
\label{eq:equation7}
R(S,v_k) =
\begin{cases}
-\min({d_{ik}}+{d_{kj}}-{d_{ij}},{d_{tk}}+{d_{k0}}- {d_{t0}}), & t_{S'}(v) \neq \text{inf}\\
-\text{inf}, & \text{otherwise} \\
\end{cases}
\end{equation}
where $d_{ij}$ is the euclidean distance between nodes $v_i$ and $v_j$, and $t_{S'}(v)$ is the updated charging time of  node in path formed by $S'$ after inserting the $v_k$.

\subsection{Deep-Q-Network (DQN) Algorithm}
We adopt the deep-Q-network algorithm introduced in~\cite{khalil2017learning} to learn the parameters $\Theta$ of the state-action value function $Q(S,v;\Theta)$. The adopted algorithm updates the parameters $\Theta$ every $n$ steps instead of every single step, to have a more accurate estimation of the objective function. Particularly, we apply experience replay method in~\cite{mnih2013playing} to update $\Theta$. Experience learned in each step is stored in a dataset. When we update $\Theta$, a sample batch with size $32$ is randomly selected from the dataset. Experience replay method helps break data correlation and avoid oscillations with the parameters. 

The input of the network is the $p$-dimension vector featured by graph embedding network and output is the optimal solution for the current state. The advantage of DQN is that it can handle delayed rewards, which represent the way to optimize the objective function. In each step of the algorithm, the graph embedding method will be used to update the current partial solution and the new $p$-dimension vector which contains the newest information will be used for the next step.

Algorithm~\ref{DQN_algorithm} summaries the major steps of the algorithm.

\begin{algorithm}
\caption{DQN Algorithm}
\begin{algorithmic}[1]
\STATE Initialize replay memory $\mathcal{H}$ to capacity $C$
\FOR{ each episode }
\STATE Initialize state $s_1 = (v_0)$
\FOR{ step $t = 1$ \TO $n$ }
\STATE Select $v_t$ by \eqref{eq:epsilon_greedy}
\STATE Add $v_t$ to partial solution by insertion function: $s_{t+1}:=g(s_t,v)$
\STATE Calculate reward $r(s_t,v)$
\STATE Store tuple $(s_{t},v,r(s_t,v),s_{t+1})$ to $\mathcal{H}$
\STATE Sample random batch $(s_l,v_l,r(s_l,v_l),s_{l+1})$ from $\mathcal{H}$
\STATE Update the network parameter $\Theta$ by squared loss function $(y_l - Q(s_l,v_l;\Theta))^2$,
        where
        \begin{equation}\label{eq:q_learning}
        y_l =
        \begin{cases}
        r(s_l,v)+\gamma\max_{v'} Q(s_{l+1},v';\Theta)\\
        \quad \quad \quad \text{if } s_{l+1} \text{ non-terminal and } v'\in \bar{s_l} \\
        r(s_l,v) \\
        \quad \quad \quad \text{otherwise}
        \end{cases}
        \end{equation}
\IF{$s_{t+1}$ satisfy the stop function}
\STATE Break
\ENDIF
\ENDFOR
\ENDFOR	
\end{algorithmic}
\label{DQN_algorithm}
\end{algorithm}

\section{Performance Comparison}
\label{Sec:Simulation}
We implement and compare the performance of DQN algorithms designed based on the proposed deep reinforcement learning-based charger scheduling optimization framework with existing ones on the selected charger scheduling optimization problems.

\subsection{Mobile Network Charging Path Optimization Problem}

We compare the DQN algorithm designed based on the proposed framework  with the quasi-polynomial time approximation algorithm (APP) introduced in~\cite{chen2016charge} and two other heuristic algorithms, including a greedy one and a random one. 

However, the computing time of the APP algorithm in~\cite{chen2016charge} is very sensitive to the  network size $n$, time step size $\Delta t$, and charging timespan $C$ and can increase dramatically to days with large  $n$, small $\Delta t$, or long  $C$. 

To control the running time, we choose networks with moderate sizes and a small timespan of charging $C=30$ minutes. We compare the computing time and the  number of sensors charged of these algorithms and  study the impact of varying network size and charger capacity on the performances in Secs.~\ref{sec:network_size1} and~\ref{sec:time_size1}, respectively. Note that we choose a low recursion level $L=3$ for the APP algorithm in~\cite{chen2016charge} to control the running time.

\subsubsection{Simulation Setting}

We deploy sensors by uniform random distribution with the size $n$ varying from $10$ to $30$ in an Euclidean square with size $[100, 100]m^2$. A charger starts from the center of monitored area and ends at the upper right corner. The speed of  charger is $5$m/s. We use the random way point model to simulate the trace of a mobile sensor node with its average speed randomly chosen from $[0, 2]$ m/s. 

The battery capacity $B$ of each sensor is $10.8$KJ \cite{chen2016charge}. The initial battery level of node $i$ is randomly chosen from $[0, B]$. $g_i[(1-\epsilon)\alpha]$ is the time to charge node $i$ to $(1-\epsilon)\alpha$.
The required charging level $\alpha = 90\%$, $\epsilon = 0.1$. The average energy transfer rate is $40$ W. 

\begin{figure*}
\begin{center}
\begin{tabular}{cc}
\includegraphics[height=5cm]{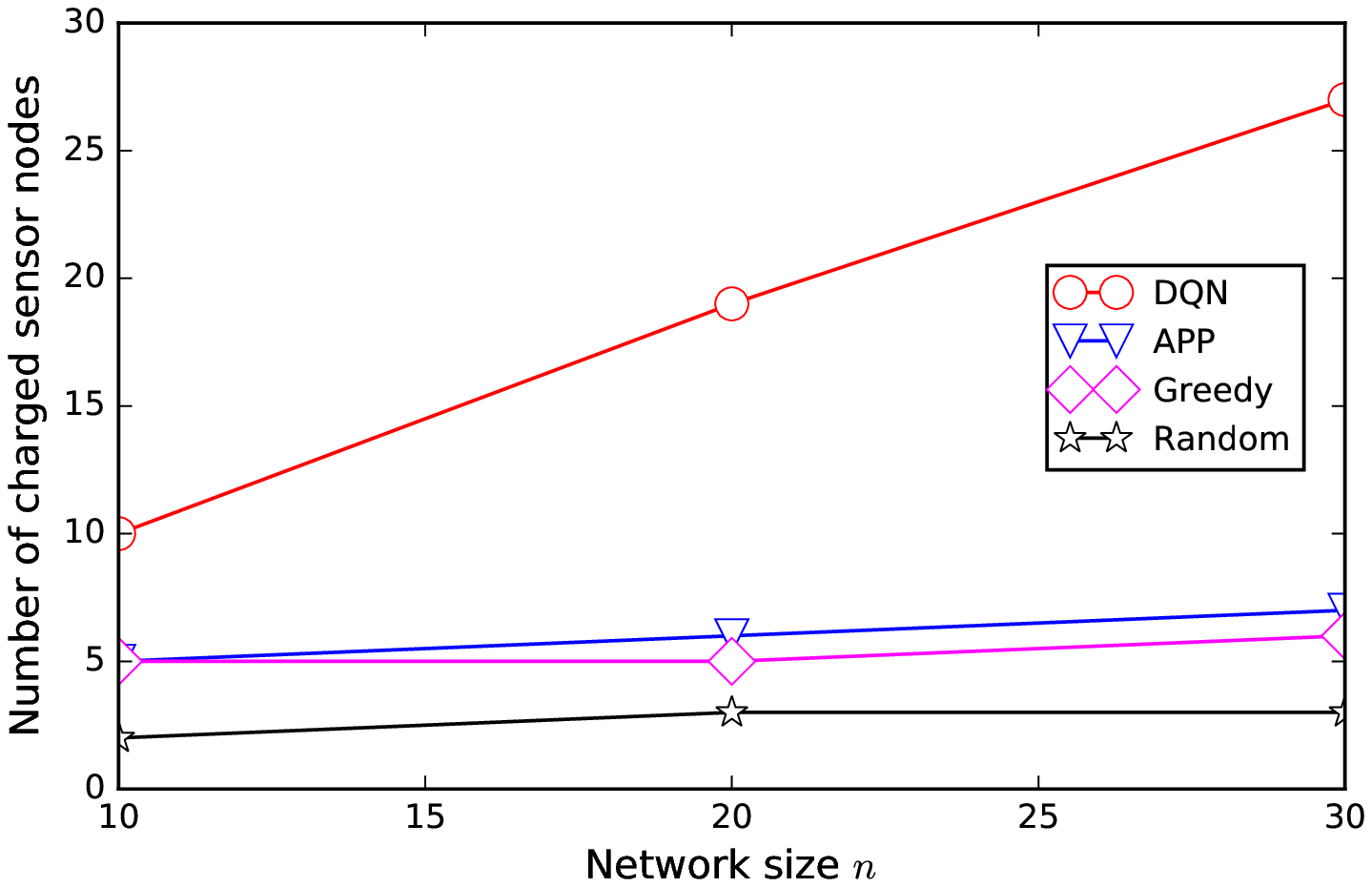} &
\includegraphics[height=5cm]{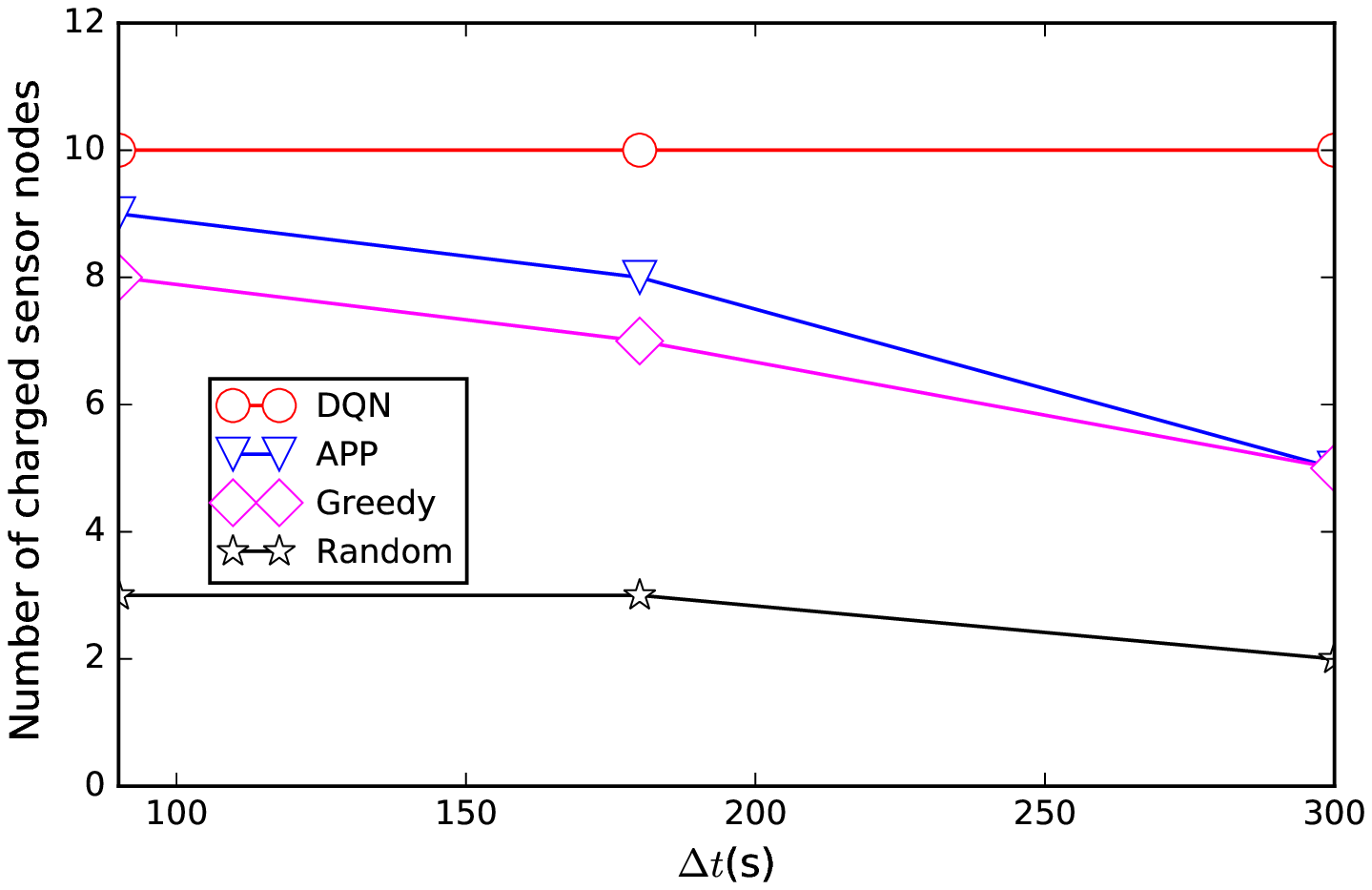} \\
(a) Impact of network size  & (b) Impact of time step size\\
\end{tabular}
\end{center}
\caption{(a) The number of  mobile sensors charged with the maximum timespan of charger $C=30 $ minutes  and $\Delta t = 300$s under a varying network size of $n$. (b) The number of  mobile sensors charged with the maximum timespan of charger $C=30 $ minutes  and $n = 10$ under a varying time step size.
\label{fig:n}}
\end{figure*}

\subsubsection{Impact of Network Size}
\label{sec:network_size1}

Figure~\ref{fig:n}(a) and Table~\ref{tab:performance} compare the number of mobile sensors charged within the maximum timespan and the corresponding computing time of each algorithm, respectively, with $\Delta t = 300$s and a varying network size of $n$. With an increased network size, more mobile sensors can be charged, but the DQN algorithm consistently charges a significantly higher number of mobile sensors than any other algorithm. At the same time, the computing time of the DQN algorithm remains stable with an increase of $n$. By contrast, the computing time of the APP algorithm in~\cite{chen2016charge} increases dramatically. Overall, the DQN algorithm outperforms all others.

\subsubsection{Impact of Time Step Size}
\label{sec:time_size1}

The time step size $\Delta t$ has a big impact on the performance of the APP algorithm in~\cite{chen2016charge} as shown in Figure~\ref{fig:n}(b) and Table \ref{tab:performance}. The number of charged mobile sensors of the APP algorithm in~\cite{chen2016charge} increases with a decreased time step size $\Delta t$ but at the cost of sky-high computing time. On the contrary, $\Delta t$ has no effect on the DQN algorithm because of the way to construct a graph as introduced in Sec.\ref{Sec:Graph_Construction}. The DQN algorithm again outperforms all other algorithms.

\begin{table}
\caption{Computing time under different network size $n$ and time stepsize $\Delta t$ \label{tab:performance}}
\begin{center}
\ra{1.3}
{    
\begin{tabular}[t]{ @{}  c | c | c | c |c  @{}} \toprule
&  \multicolumn{2}{c |}{Computing time in Fig. \ref{fig:n}(a)}  & \multicolumn{2}{c}{Computing time in Fig. \ref{fig:n}(b)}  \\
\cline{2-3}	\cline{4-5}
Alg. & $n$ & Comp. Time& $\Delta t$  & Comp. Time \\
&      & (s)   & (s)     &  (s) \\ \hline	
DQN  & \multirow{5}{*}{$10$} &  22 & \multirow{5}{*}{$90$}&  22 \\
APP  &  &  580  & &  417329\\
Greedy & & 0.01 & &  0.033\\
Random & & 0.0006 & &  0.0007\\\hline
DQN  & \multirow{5}{*}{$20$} &  53 &\multirow{5}{*}{$180$} & 22\\
APP  &  &  28011 &  & 7507\\
Greedy & &  0.044 & & 0.02\\
Random & &  0.0024 & & 0.0006\\\hline
DQN  & \multirow{5}{*}{$30$} &  70 & \multirow{5}{*}{$300$} &  22\\
APP  &  &  232080 &  &  580\\	
Greedy & & 0.11 & & 0.01\\
Random & & 0.005 & & 0.0006\\
\bottomrule
\end{tabular}
}
\end{center}
\end{table}

\subsection{Fully Charging Reward Maximization Problem}\label{Sec:Simulation2}

Paper~\cite{liang2017approximation} provides a $4$-approximation algorithm to solve the fully charging reward maximization problem by reducing the original problem to a classical orienteering one~\cite{bansal2004approximation}.

We compare the DQN algorithm designed based on the proposed framework with the approximation algorithm (APP) in~\cite{liang2017approximation}, the minimum spanning tree (MST),  the Capacitated Minimum Spanning Tree (CMST), and the greedy one. We compare the computing time and the total energy spent on charging sensors and  study the impact of varying network size and charger capacity on the performances of these algorithms in Secs.~\ref{sec:network_size2} and~\ref{sec:capacity}, respectively.

\subsubsection{Simulation Setting}

The simulation area is an Euclidean square $[1000, 1000]m^2$. We randomly deploy sensors with size $n$ ranging from $50$ to $200$. A mobile charger with energy capacity $IE = 300$KJ, travels with an average speed $5$m/s from the center of monitored area and back to it with an average $600$ J/m spent on traveling~\cite{liang2017approximation}. The battery capacity $B$ of each sensor is $10.8$KJ \cite{liang2017approximation}. The initial battery level of node $v_i$ is randomly chosen from $(0, B]$. A sensor node sends a charging request when its energy is below $20\%$ of its capacity $B$. 

\subsubsection{Impact of Network Size}
\label{sec:network_size2}

Figure~\ref{fig:fully_charging}(a) and Table \ref{tab:fully_charging} compare the energy spent on sensor charging and the corresponding computing time of each algorithm, respectively, with the network size $n$ increased from $50$ to $200$. It is clear that the DQN algorithm spends much more energy in charging sensors than any other algorithm. Considering a charger with a fixed energy capacity and a group of sensors scattered in a field,  the energy spent on charging sensors decreases with an increased network size for all algorithms. However, with the increased network size,  the computing time of the APP algorithm in~\cite{liang2017approximation}  increases dramatically while the DQN algorithm remains stable. Overall, the DQN algorithm outperforms all other algorithms.

\subsubsection{Impact of Charger Capacity}
\label{sec:capacity}

Figure~\ref{fig:fully_charging}(b) and Table \ref{tab:fully_charging} compare the energy spent on sensor charging and the corresponding computing time of each algorithm, respectively, with a varying charger capacity.  With the charger capacity $IE$ increased from $200$KJ to $350$KJ, the energy spent on charging sensor nodes increases for all algorithms too. The DQN algorithm consistently charges much more energy on sensors than any other algorithms. At the same time, the computing time of the DQN algorithm keeps stable while the APP algorithm increases dramatically.  Overall, the DQN algorithm still outperforms all other algorithms.

\begin{figure*}
\begin{center}
\begin{tabular}{cc}
\includegraphics[height=5cm]{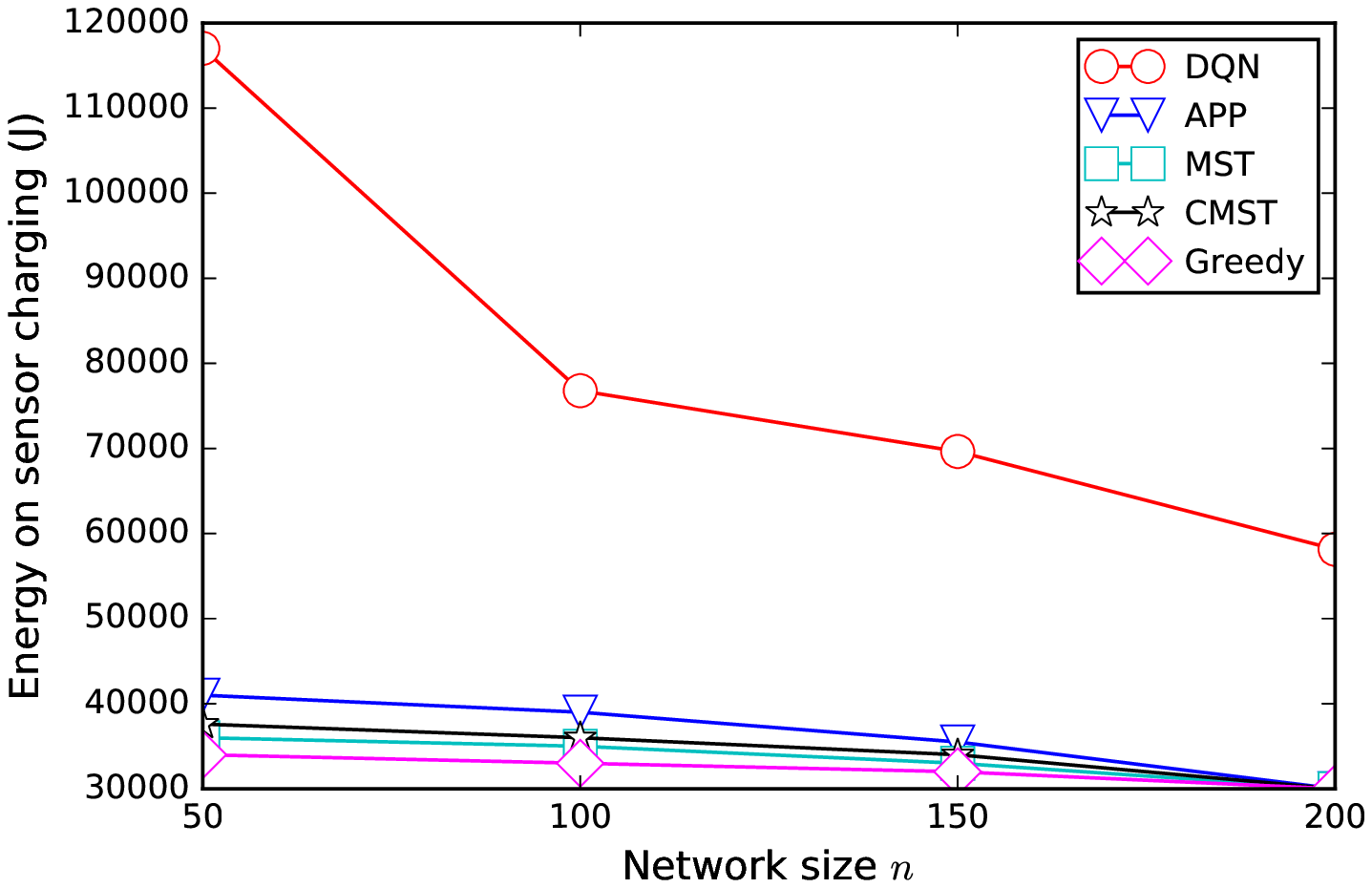} &
\includegraphics[height=5cm]{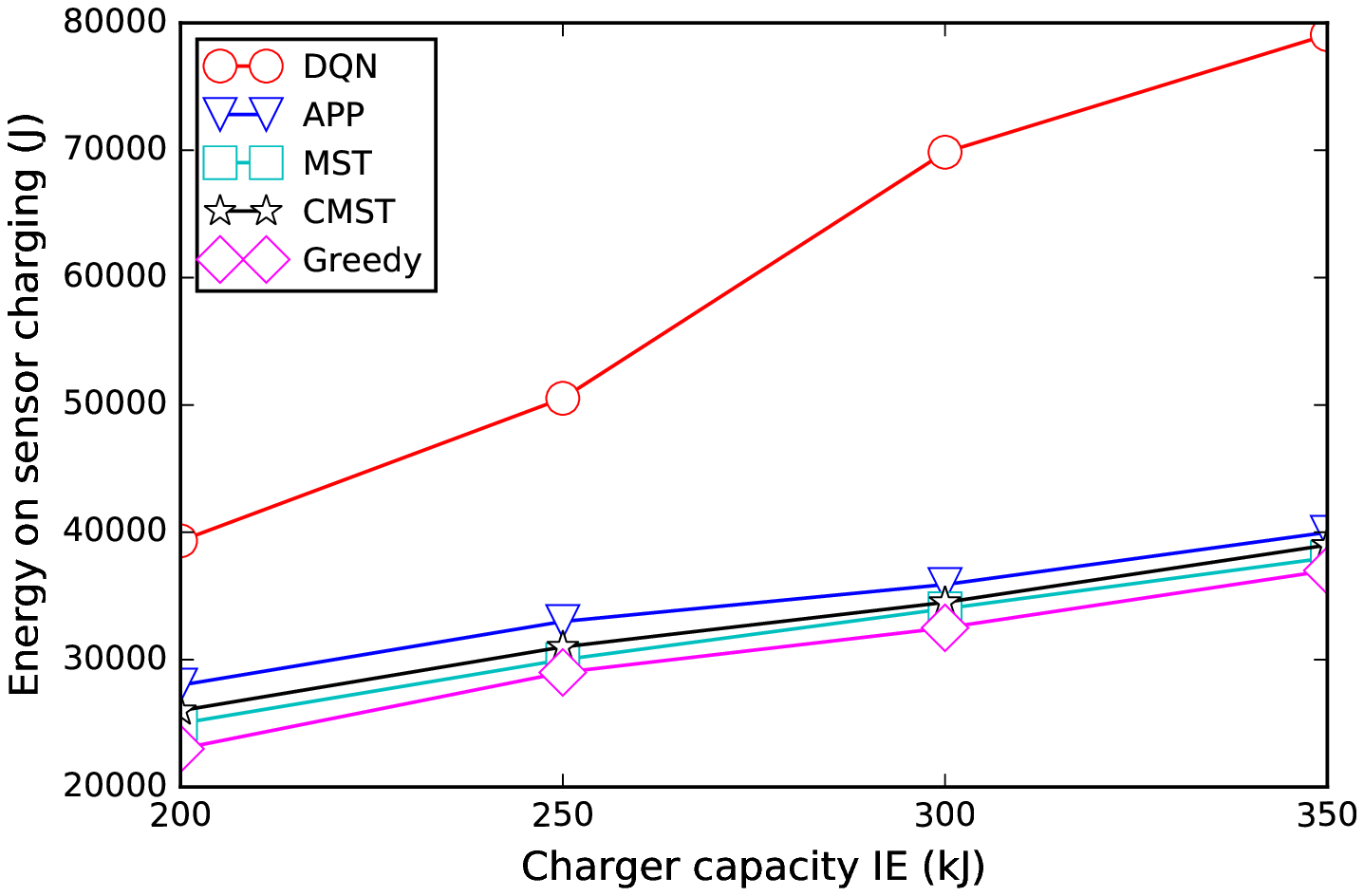} \\
(a) Impact of network size & (b) Impact of charger capacity \\
\end{tabular}
\end{center}
\caption{ (a) Energy spent on sensor charging  for a mobile charger with a total energy capacity $IE = 300$KJ under a varying network size of $n$. (b) Energy spent on sensor charging for a mobile charger with a varying energy capacity $IE$ and a fixed network size $n=120$.
\label{fig:fully_charging}}
\end{figure*}

\begin{table}
\caption{Computing time under different network size $n$ and charger capacity $IE$  \label{tab:fully_charging}}
\begin{center}
    \ra{1.3}
{    
\begin{tabular}[t]{ @{}  c | c | c | c |c  @{}} \toprule
            &  \multicolumn{2}{c |}{Computing time in Fig. \ref{fig:fully_charging}(a)}  & \multicolumn{2}{c}{Computing time in Fig. \ref{fig:fully_charging}(b)}  \\
                \cline{2-3}	\cline{4-5}
                Alg. & $n$   & Comp. Time & $IE$     & Comp. Time \\
                &       & (s)        & (kJ)     &  (s) \\ \hline	
DQN  & \multirow{5}{*}{$50$} &  31  & \multirow{5}{*}{$200$} &  67\\
APP  &  &  35  &  &  90\\
                MST  &  &  0.2  &  &  0.39\\
                CMST  &  &  1.2  &  &  2.35\\	
Greedy & & 0.0005 & & 0.0004\\\hline
                DQN  & \multirow{5}{*}{$100$} &  57 & \multirow{5}{*}{$250$} &  67\\
APP  &  &  506  &  &  102\\
                MST  &  &  0.35  &  &  0.4\\
                CMST  &  &  2.45  &  &  2.37\\	
Greedy & & 0.0007 & & 0.0004\\\hline
                DQN  & \multirow{5}{*}{$150$} &  87  & \multirow{5}{*}{$300$} &  68\\
APP  &  &  1005  &  &  605\\
                MST  &  &  0.46  &  &  0.42\\
                CMST  &  &  3.56  &  &  2.95\\	
Greedy & & 0.0009 & & 0.0004\\\hline
                DQN  & \multirow{5}{*}{$200$} &  115  & \multirow{5}{*}{$350$} &  70\\
APP  &  &  3500  &  &  807\\
                MST  &  &  0.58  &  &  0.49\\
                CMST  &  & 4.28  &  &  2.97\\	
Greedy & & 0.001 &  & 0.0004\\
            \bottomrule
\end{tabular}
}
\end{center}
\end{table}

\subsection{Optimal k-coverage Charging Problem}\label{Sec:Simulation3}

We evaluate the DQN algorithm designed based on the proposed framework and compare with the dynamic programming based algorithm and three other heuristic algorithms including Ant Colony System (ACS) based algorithm, Random algorithm, and Greedy algorithm, as explained briefly in Sec.~\ref{Sec:Comparison_Algorithms}. 

We compare the computing time to find a feasible charging path and the energy of a charger spent on traveling of these algorithms. Considering network settings may affect the performance, we study the impact of varying network size $n$, coverage requirement $k$, and remaining energy threshold $\alpha$  in Secs.~\ref{Sec:Size}, ~\ref{Sec:K}, and~\ref{Sec:Percentage}, respectively. Specifically, we assume a monitored area is at least $k$-coverage initially. Sensor node $v_i$ estimates the charging deadline $D_i$  based on its residual energy $B_i(t_0)$ and the experimental energy consumption rate in~\cite{Zhu:2009:LES:1555816.1555849}. 

\subsubsection{Simulation Setting}

We set up an Euclidean square $[500, 500]m^2$ as a simulation area and randomly deploy sensor nodes with size ranging from $32$ to $80$ in the square such that the area is at least $k$-coverage initially where $k$ varies from $2$ to $4$. The sensing range $r$ is $135$m. The base  and service station of charger are co-located in the center of the square. A charger with a starting point from the service station has an average traveling speed $5$m/s and consumes energy $600J/m$ \cite{liang2017approximation}. The battery capacity $B$ of each sensor is $10.8$KJ \cite{liang2017approximation}. The remaining energy threshold $\alpha$ vary from $0.2$ to $0.8$. A sensor sends a charging request before the leaving of the charger from the service station. The sensor will include in the request the estimated energy exhausted time based on its current residual energy and energy consumption rate. To simulate such request, we consider the residual battery of a sensor is a uniform random variable $B_i$ between $(0.54,10.8]$KJ \cite{shi2011renewable}. The energy exhausted time $D_i = B_i/\beta_i$. We choose the energy consumption rate $\beta_i$ from the historical record of real sensors in~\cite{Zhu:2009:LES:1555816.1555849} where the rate is varying according to the remaining energy and arrival charging time to the sensor. The energy transfer rate $r_c$ is $20$W \cite{chen2016charge}. The discredited time step-size is $1$s.

\subsubsection{Comparison Algorithms}\label{Sec:Comparison_Algorithms}

We implement three  heuristic algorithms for comparison: Ant Colony System (ACS) based algorithm, Random algorithm, and Greedy algorithm.

ACS algorithm solves the traveling salesmen problem with an approach similar to the foraging behavior of real ants~\cite{ant1991,dorigo1997ant,Gutjahr:2000:GAS:348599.348602}. Ants seek path from their nest to food sources and leave a chemical substance called pheromone along the paths they traverse. Later ants sense the pheromone left by earlier ones and tend to follow a trail with a stronger pheromone. Over a period of time, the shorter paths between the nest and food sources are likely to be traveled more often than the longer ones. Therefore, shorter paths accumulate more pheromone, reinforcing these paths.

Similarly, ACS algorithm places agents at some vertices of a graph. Each agent performs a series of random moves from current vertex to a neighboring one based on the transition probability of the connecting edge. After an agent has finished its tour, the length of tour is calculated and the local pheromone amounts of edges along the tour are updated based on the quality of the tour. After all agents have finished their tours, the shortest one is chosen and the global pheromone amounts of edges along the tour are updated. The procedure continues until certain criteria are satisfied. When applying ACS algorithm to solve the the traveling salesmen problem with deadline, two local heuristic functions are introduced in~\cite{cheng2007modified} to exclude paths that violate the deadline constraints.

We modify ACS algorithm introduced in~\cite{cheng2007modified} for the optimal $k$-coverage charging problem. Agents start from and end at $v_0$.  Denote $\tau_{ij}(t)$ the amount of global pheromone deposited on edge $\overrightarrow{v_iv_j}$ and $\Delta\tau_{ij}(t)$ the increased  amount at the $t^{th}$ iteration. $\Delta\tau_{ij}(t)$ is defined as
\begin{equation}
\label{eq:Deltatau_ij}
\Delta\tau_{ij}(t) =
\begin{cases}
\frac{1}{L^*}  & \text{if } \overrightarrow{v_iv_j} \in P^*\\
0 & \text{otherwise} \\
\end{cases}
\end{equation}
where $L^*$ is the traveling distance of the shortest feasible tour $P^*$ at the $t^{th}$ iteration. Global pheromone $\tau_{ij}(t)$ is updated according to the following equation:
\begin{equation}
\label{eq:tau_ij}
\tau_{ij}(t) = (1-\theta)\tau_{ij}(t-1) + \theta\Delta\tau_{ij}(t),
\end{equation}
where $\theta$ is the global pheromone decay parameter. The local pheromone is updated in a similar way, where $\theta$ is replaced by a local pheromone decay parameter and $\Delta\tau_{ij}(t)$ is set as the initial pheromone value.

We also modify the stop criteria of one agent such that the traveling path satisfies the requirement of $k$-coverage, or the traveling time of current path is inf, or the agent is stuck at a vertex based on the transition rule.

Random and Greedy algorithms work much more straightforward. The Random algorithm randomly chooses a next node not in the same clique of the existing path and with an outgoing edge from current one to charge. The Greedy algorithm always chooses the nearest node not in the same clique of the existing path with an outgoing edge from current one. Random and Greedy algorithms terminate when they either find a feasible path or are locally stuck. Note that for ACS and Random algorithms, we always run multiple times and choose the best solution.

\subsubsection{Impact of Network Size}\label{Sec:Size}
 
 We set the coverage requirement $k=3$ and the remaining energy threshold $\alpha = 0.45$ for a sensor node to send a charging request. Table~\ref{tab:parameter_n_result} compares the performances including the computation time to find a feasible charging path and the energy spent on traveling when the network size $n$ varies from $48$ to $80$. The energy spent on traveling decreases with an increased $n$ since there are more redundant sensors to maintain the $k$-coverage requirement. Again, the DQN algorithm outperforms all other competing  algorithms. 
 
 \begin{table}
 	\caption{Performance comparison under different sizes of sensor network $n$ \label{tab:parameter_n_result}}
 	\begin{center}
 		\ra{1.3}
 		{    \fontsize{9pt}{9pt}\selectfont
 			\begin{tabular}[t]{ @{}  c | c | c | c |c  @{}} \toprule
 				& & \multicolumn{3}{c }{$k=3$, $\alpha =0.45$}  \\ \cline{3-5}
 				Algorithm & $n$  & Computation & Feasible  & Traveling \\
 				&   & Time  &  Path &  Energy \\
 				&   & (s) & Found  &   (kJ)  \\ \hline	
 				Dynamic  & \multirow{5}{*}{48} & 156700 & Yes & 771 \\
 				DQN  &&   83 & Yes & 771   \\
 				ACS  &&  92 & Yes & 951  \\
 				Random & & 0.0004 & Yes & 2085\\
 				Greedy & & 0.0004 & Yes & 888\\ \hline
 				Dynamic  & \multirow{5}{*}{64} & 455 & Yes    & 702  \\
 				DQN   &  & 20 & Yes    & 702  \\
 				ACS   &  & 19 & Yes    & 702  \\
 				Random & & 0.0003 & No     & --   \\
 				Greedy & & 0.0003 & Yes     & 846   \\ \hline
 				Dynamic   & \multirow{5}{*}{72} & 362 & Yes & 567 \\
 				DQN   & & 32 & Yes & 567\\
 				ACS   & & 57 & Yes & 567 \\
 				Random & & 0.0003 & Yes & 1941 \\
 				Greedy & & 0.0003 & Yes & 810 \\ \hline
 				Dynamic   & \multirow{5}{*}{80} & 268 & Yes & 345 \\
 				DQN   & & 20 & Yes & 345\\
 				ACS   & & 18 & Yes & 345\\
 				Random & & 0.0003 & No & -- \\
 				Greedy & & 0.0003 & Yes & 375\\
 				\bottomrule
 			\end{tabular}
 		}
 	\end{center}
 \end{table}

\subsubsection{Impact of Coverage Requirement} \label{Sec:K}

We set the number of sensor nodes $n= 64$ and the remaining energy threshold $\alpha = 0.45$ for a sensor node to send a charging request. Table~\ref{tab:parameter_k_result} compares the performances including the computation time to find a feasible charging path and the energy spent on traveling when the coverage requirement $k$ varies from $2$ to $4$. The traveling energy increases with the increased $k$ because more sensor nodes need to be charged to satisfy the coverage requirement. The dynamic programming algorithm with an exponentially increased computing time can only find a feasible charging path when $k$ is small.  By contrast, the computing time of the DQN algorithm including its training time grows slowly with the increase of $k$. The ACS algorithm performs better than the random and greedy algorithms, but overall, the performance of the DQN algorithm significantly outperforms all other competing  algorithms.

\begin{table}
	\caption{Performance comparison under different coverage requirement $k$\label{tab:parameter_k_result}}
	\begin{center}
		\ra{1.3}
		{    \fontsize{9pt}{9pt}\selectfont
			\begin{tabular}[t]{ @{}  c | c | c | c |c  @{}} \toprule
				&   & \multicolumn{3}{c}{$n=64$, $\alpha = 0.45$} \\ \cline{3-5}	
				Algorithm & $k$  & Computation& Feasible & Traveling \\
				&   & Time  &  Path &  Energy \\
				&   & (s) & Found  &   (kJ) \\ \hline	
				Dynamic  & \multirow{5}{*}{2} & 0.102 & Yes    & 249 \\
				DQN  &  & 16 & Yes    & 249 \\
				ACS  & &  5  & Yes    & 249 \\
				Random & & 0.0006 & Yes    & 249  \\
				Greedy & & 0.0007 & Yes    & 288 \\ \hline
				Dynamic  & \multirow{5}{*}{3} & 455 & Yes    & 702  \\
				DQN   &  & 20 & Yes    & 702  \\
				ACS   &  & 19 & Yes    & 702  \\
				Random & & 0.0003 & No     & --   \\
				Greedy & & 0.0003 & Yes     & 846   \\ \hline
				Dynamic   & \multirow{5}{*}{4} & -- & --    & -- \\
				DQN   &  & 71 & Yes    & 1089 \\
				ACS   &  & 73 & Yes     & 1188  \\
				Random & & 0.0006 & No     & --  \\
				Greedy & & 0.0005 & Yes     & 1254  \\
				\bottomrule
			\end{tabular}
		}
	\end{center}
\end{table}

\subsubsection{Impact of Remaining Energy Threshold}\label{Sec:Percentage}

Tables~\ref{tab:parameter_alpha_result_1} and~\ref{tab:parameter_alpha_result_2} give the performances of different algorithms with the remaining energy threshold $\alpha$ varying from $0.2$ to $0.8$ under two network settings: $n=32$ and $k=2$, and $n=48$ and $k=3$, respectively. With the increased remaining energy threshold, the traveling energy increases in both network settings. The Random and Greedy algorithms fail to detect a feasible path in many cases. The dynamic programming algorithm runs out of the memory when $\alpha$ is large. Both the ACS and DQN algorithms find feasible paths for all cases. However, the performance of the ACS algorithm decreases with the increase of $\alpha$. The DQN algorithm consistently and significantly outperforms all other comparison algorithms including botht the traveling energy and computing time.

\begin{table}
	\caption{Performance comparison under different remaining energy threshold $\alpha$ when $k=2$, $n = 32$\label{tab:parameter_alpha_result_1}}
	\begin{center}
		\ra{1.3}
		{    \fontsize{9pt}{9pt}\selectfont
			\begin{tabular}[t]{ @{}  c | c | c | c | c @{}} \toprule
				&   & \multicolumn{3}{c}{$k=2$, $n = 32$} \\ 	\cline{3-5}			
				Algorithm & $\alpha$  & Computation & Feasible  & Traveling  \\
				&   & Time  &  Path &  Energy \\
				&   & (s) & Found  &   (kJ)  \\ \hline	
				Dynamic  & \multirow{5}{*}{0.2} & 0.0012 & Yes    & 405\\
				DQN    && 13&  Yes    & 405\\
				ACS    && 3&  Yes    & 405\\
				Random && 0.0002&  No    & --\\
				Greedy && 0.0003&  Yes     & 420\\ \hline
				
				Dynamic  & \multirow{5}{*}{0.4} & 9760&  Yes    & 696\\
				DQN    && 26&  Yes    & 696  \\
				ACS    & & 38 &  Yes    & 855 \\
				Random && 0.0003 & No     & --   \\
				Greedy & & 0.0003&  No     & 891 \\ \hline
				
				Dynamic  &  \multirow{5}{*}{0.6} & 135742 & Yes    & 1071  \\
				DQN    && 116	&  Yes    & 1071  \\
				ACS    & & 232&  Yes    & 2289  \\
				Random & & 0.0006&  No     & --   \\
				Greedy && 0.0004 &  No     & --   \\ \hline
				
				Dynamic  & \multirow{5}{*}{0.8} & -- &  --        & -- \\
				DQN    && 136 & Yes        & 1080 \\
				ACS    && 400 &  Yes         & 2544    \\
				Random && 0.002 &  No         & --   \\
				Greedy & & 0.001& No        & --   \\
				\bottomrule
			\end{tabular}
		}
	\end{center}
\end{table}

\begin{table}
	\caption{Performance comparison under different remaining energy threshold $\alpha$ when $k=3$, $n = 48$\label{tab:parameter_alpha_result_2}}
	\begin{center}
		\ra{1.3}
		{   \fontsize{9pt}{9pt}\selectfont
			\begin{tabular}[t]{ @{}  c | c | c | c | c @{}} \toprule
				&   & \multicolumn{3}{c }{$k=3$, $n = 48$}  \\ 	\cline{3-5}			
				Algorithm & $\alpha$  & Computation & Feasible  & Traveling \\
				&   & Time  &  Path &  Energy \\
				&   & (s) & Found  &   (kJ) \\ \hline	
				Dynamic  & \multirow{5}{*}{0.2} & 0.02 & Yes    & 348 \\
				DQN    && 10 & Yes    & 348 \\
				ACS    && 1 & Yes    & 348 \\
				Random && 0.0003 & No    & -- \\
				Greedy && 0.0002 & Yes    & 348 \\ \hline
				
				Dynamic  & \multirow{5}{*}{0.4} & 145602& Yes & 750 \\
				DQN  &&   81 & Yes & 750   \\
				ACS  &&  90 & Yes & 930  \\
				Random & & 0.0004 & Yes & 2070\\
				Greedy & & 0.0004 & Yes & 870\\ \hline
				
				Dynamic  &  \multirow{5}{*}{0.6} & -- & --    & --\\
				DQN    && 138 & Yes    & 1020\\
				ACS    && 466 & Yes    & 1521   \\
				Random & & 0.001 & No    & --   \\
				Greedy && 0.001 & No    & --   \\ \hline
				
				Dynamic  & \multirow{5}{*}{0.8} & --  & --    & --   \\
				DQN    && 481 & Yes    & 1272   \\
				ACS    && 4200   & Yes    & 4788   \\
				Random && 0.01 & No    & --   \\
				Greedy && 0.02 & No    & --   \\
				\bottomrule
			\end{tabular}
		}
	\end{center}
\end{table}

\section{Conclusions}\label{Sec:Conclusion}

We introduce a deep reinforcement learning-based charger scheduling optimization framework. The biggest advantage of the framework is that a diverse range of domain-specific charger scheduling strategy can be learned automatically from previous experiences, i.e., different graphs with various sizes. A framework also simplifies the complexity of algorithm design for individual charger scheduling optimization problem. We compare the performance of algorithms designed based on the proposed framework with existing ones on a set of representative charger scheduling optimization problems. Extensive simulation and comparison results show that algorithms based on the proposed framework outperform all existing ones and achieve close to optimal results.

\bibliographystyle{unsrtnat}
\bibliography{bio} 

\end{document}